\documentclass[12pt,draftcls,onecolumn]{IEEEtran}

\IEEEoverridecommandlockouts                              
\usepackage{amssymb,latexsym,color,amsmath,pifont,colordvi,multicol}
\usepackage{times}

\definecolor{backgrey}{rgb}{0.86,0.86,0.86}
\definecolor{dblue}{rgb}{0,0.0,0.5}
\definecolor{dred}{rgb}{0.4,0.2,0}
\definecolor{dgreen}{rgb}{0.0,0.5,0}

\newcommand{\captionfonts}{\small}
\makeatletter  
\long\def\@makecaption#1#2{%
  \vskip\abovecaptionskip
  \sbox\@tempboxa{{\captionfonts #1: #2}}%
  \ifdim \wd\@tempboxa >\hsize
    {\captionfonts #1: #2\par}
  \else
    \hbox to\hsize{\hfil\box\@tempboxa\hfil}%
  \fi
  \vskip\belowcaptionskip}
\makeatother   

\usepackage{graphics} 
\usepackage[pdftex]{graphicx}

\newtheorem{theorem}{Theorem}

\newtheorem{remark}{Remark}
\newtheorem{proof}{Proof}

\newtheorem{definition}{Definition}

\newtheorem{lemma}{Lemma}

\title{\LARGE \bf Information Flow Decomposition in Feedback Systems: General Case Study}

\author{\quad Bertrand Wechsler, Dan Eilat and Nicolas Limal
}

\begin{document}
\maketitle \thispagestyle{empty} \pagestyle{plain}
\begin{abstract}
We derive three fundamental decompositions on relevant information quantities in feedback systems. The feedback systems considered in this paper are only restricted to be causal in time domain and the channels are allowed to be subject to arbitrary distribution. These decompositions comprise the well-known mutual information and the directed information, and indicate a law of conservation of information flows in the closed-loop network.
\end{abstract}

\section{Introduction}

\indent Feedback systems have been well studied and understood in the community of control theory since almost one hundred years ago. In 1960's, the communication community started to pay attention to feedback systems and, since then, many notable results have been established \cite{Shannon58,Schalkwijk66_2,Cover88, bookhan03, chong_isit11,Kim08_capacity_fb, Kim10,Chong11_isit_bounds, Chong11_allerton_UpperBound,Ofer07,Tati09,Permuter09}. One recent breakthrough is the notion of directed information introduced by Massey \cite{Massey1990}. The directed information successfully assesses the amount of information flowing from one random sequence to another in a causal fashion. This notion has wide applications in different research fields. For example, it characterizes the capacity of channels with noiseless feedback \cite{Kim08_capacity_fb, Kim10,Ofer07,Tati09,Permuter09}, provides understanding on portfolio theory, data compression and hypothesis testing \cite{Permuter_portfolio}, and develops fundamental limitations for networked control system \cite{Martins_control}.\\
\indent The relationship between the directed information and mutual information has been well investigated and a conservation law has been found in \cite{Massey_conservation}. However, the relationship between these two important quantities in general feedback system (as shown in Fig. \ref{Fig:general_feedback_system}) is not clear until now. In Fig. \ref{Fig:general_feedback_system}, $S_1$ and $S_2$ are two plants/systems, communicating to each other through noisy channels represented by $C_1$ and $C_2$. $m$ is some external information given to the system $S_1$ and $z$ is the information of $m$ extracted from system $S_2$. When specified to feedback communication systems, $m$ is the message and $E$ represents the encoder/transmitter. $D$ represents the decoder and $\hat{m}$ is the estimated message. For the ease of readability, in this paper, we are interested in the case of $u = y$ (i.e., the system $S_2$ is an unit gain system), as shown in Fig. \ref{Fig:general_feedback_system} (right). The results can be naturally extended to the general case as shown in Fig \ref{Fig:general_feedback_system} (left).\\
\indent In the literature, there exist a few results on feedback communication systems, to list a few \cite{Chong11_allerton_finiteCapacity, Draper06, Kim07,Chong12_allerton_sideInfo,Martins08,  Chance10}. However, it is still far away from understanding the information flow and finding the relationships among relevant quantities in general feedback systems. In this paper, we derive some connections among relevant information quantities, in order to understand the information flow in the closed-loop system.\\
\indent In what follows, we introduce the mathematical model of the feedback system considered in this paper. Without loss of generality, channel $C_1$ is modeled as
\begin{equation*}
p(y_i| x^i, y^{i-1})
\end{equation*}
where $x_i, y_i$ are respectively channel input and output at time instant $i$. $x^i$ represents the sequence $x_1, x_2, \cdots, x_i$ and $y^{i-1}$ represents the sequence $y_1, y_2, \cdots, y_{i-1}$. This probabilistic channel model indicates that the i-th channel output depends on the current channel input and all the previous channel inputs, channel outputs. Moreover, the channel input $x_i$ is determined by the message $m$, inputs $e^{i-1}$ to the system $S_1$ and previous channel inputs $x^{i-1}$. Similarly,  channel $C_2$ is modeled as
\begin{equation*}
p(e_i| e^{i-1}, y^i)
\end{equation*}
where the current channel output $e_i$ depends on the current feedback input $y_i$ and all the previous feedback inputs and outputs.



\begin{figure}
\begin{center}
\includegraphics[scale=0.70]{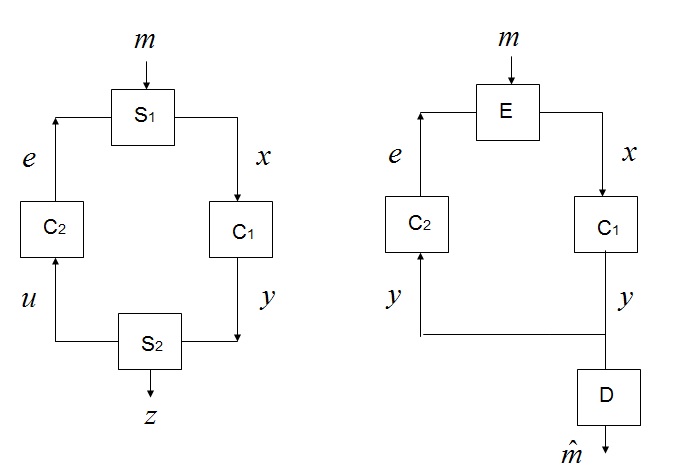}
\caption{A general feedback system (left) and a feedback communication system (right). }
\label{Fig:general_feedback_system}
\end{center}
\end{figure}

\section{Information Identities in Feedback Systems}

\indent First of all, we revisit the definition of directed information which will be repeatedly used in the paper.
\begin{definition}
Given random sequences $x^n$, $y^n$, the directed information from $x^n$ to $y^n$ is defined as
\begin{equation*}
I(x^n \rightarrow y^n) = \sum_{i=1}^{n}I(x^i;y_i|y^{i-1}).
\end{equation*}
\end{definition}

Next, we define the causal conditioning directed information \cite{chong_arXiv11}.
\begin{definition}
Given random sequences $x^n$, $y^n$ and $z^n$, the directed information from $x^n$ to $y^n$ causal conditioning on $z^n$ is defined as
\begin{equation*}
I(x^n \rightarrow y^n||z^n) = \sum_{i=1}^{n}I(x^i;y_i|y^{i-1}, z^i).
\end{equation*}
\end{definition}

Before moving forward to our main results, we present some technical lemmas below.
\begin{lemma}
Consider a feedback system as shown in Fig. \ref{Fig:general_feedback_system}. Let $m = x_0$, then the mutual information between the information injected into system $S_1$ and the information sequence $e^n$ can be characterized by
\begin{equation*}
I(x_0; e^n) = I(x^{n}\rightarrow e^n).
\end{equation*}
\label{lem01}
\end{lemma}

\begin{proof}
\begin{equation*}
\begin{split}
I(x_0;e^n)=&H(e^n)-H(e^n|x_0)\\
=&\sum_{i=1}^{n}H(e_i|e^{i-1})-\sum_{i=1}^{n}H(e_i|e^{i-1},x_0)\\
\stackrel{(a)}=&\sum_{i=1}^{n}H(e_i|e^{i-1})-\sum_{i=1}^{n}H(e_i|e^{i-1},x_0, x^i)\\
\stackrel{(b)}=&\sum_{i=1}^{n}H(e_i|e^{i-1})-\sum_{i=1}^{n}H(e_i|e^{i-1}, x^i)\\
=&\sum_{i=1}^{n}I(x^i, e_i|e^{i-1})\\
=& I(x^n\rightarrow e^n)\\
\end{split}
\end{equation*}
where (a) follows from the fact that $x^i$ is determined by $x_0$ and $e^{i-1}$. Line (b) follows from the Markov chain $x_0 - (e^{i-1}, x^i) - e_i$. This Markov chain is true because the information of $x_0$ in $e_i$ (up to time instant $i$) can only be obtained through $(e^{i-1}, x^i)$.
\end{proof}

This lemma indicates that the information $x_0$ can be leant from $e^n$ in a causal manner via the sequence $x^n$. If $x_0$ is assumed to be a message index, $S_1$ is an encoder/transmitter, and there exists a decoder/receiver taking the outputs of channel $C_2$, this lemma turns out to be the essence of capacity characterization of communication channels with noiseless feedback.

\begin{lemma}
Consider a feedback system as shown in Fig. \ref{Fig:general_feedback_system}. Let $m = x_0$, the information flow through channel $C_2$ can be decomposed into two independent flows as
\begin{equation*}
I(y^{n}\rightarrow e^n) = I(e^n; x_0) + I(y^n\rightarrow e^n|x_0).
\end{equation*}
\label{lem02}
\end{lemma}

\begin{proof}
\begin{equation*}
\begin{split}
I(y^n\rightarrow e^n)=&\sum_{i=1}^{n}I(y^i, e_i|e^{i-1})\\
=&\sum_{i=1}^{n}H(e_i|e^{i-1})-H(e_i|e^{i-1}, y^i)\\
\stackrel{(a)}=&\sum_{i=1}^{n}H(e_i|e^{i-1})-H(e_i|e^{i-1}, y^i, x_0)\\
=&\sum_{i=1}^{n}H(e_i|e^{i-1})-H(e_i| e^{i-1}, x_0) + H(e_i| e^{i-1}, x_0) - H(e_i|e^{i-1}, y^i, x_0)\\
=&\sum_{i=1}^{n}I(e_i; x_0| e^{i-1}) + I(e_i; y^i|e^{i-1},x_0)\\
=&I(x_0; e^n) + I(y^n\rightarrow e^n|x_0)\\
\end{split}
\end{equation*}
where (a) is true due to the causality of the feedback channel.
\end{proof}

Based on the above two lemmas, we have the following decomposition equality whose proof directly follows from Lemma \ref{lem01} and \ref{lem02}.

\begin{theorem}
Consider a feedback system as shown in Fig. \ref{Fig:general_feedback_system}. Let $m = x_0$, the information flow through channel $C_2$ can be decomposed into two independent flows as
\begin{equation*}
I(y^{n}\rightarrow e^n) = I(x^{n}\rightarrow e^n) + I(y^n\rightarrow e^n|x_0).
\end{equation*}
\end{theorem}

\begin{remark}
This equality can be interpreted as a law of conservation of information flows. Quantity $I(x^{n}\rightarrow e^n)$ is the amount of information provided by the external input $x_0$, and the  quantity $I(y^n\rightarrow e^n|x_0)$ is the amount of information provided by the uncertainty in the channel $C_1$ (due to the presence of noise). The sum of these two quantities equals to the total amount of information delivered from system $S_2$ to system $S_1$ through channel $C_2$.
\end{remark}
\begin{remark}
This theorem can be alternatively proved by taking Theorem $5$ and Theorem $6$ in \cite{derpich2013fundamental}. Note that a general framework of feedback systems has been investigated in \cite{derpich2013fundamental}.
\end{remark}

\begin{theorem}
Consider a feedback system as shown in Fig. \ref{Fig:general_feedback_system}. Let $m = x_0$,
\begin{equation*}
I(x^{n}\rightarrow y^n) \geq I(x^{n-1}\rightarrow e^{n-1}) + I(e^{n-1}\rightarrow y^n).
\end{equation*}
\end{theorem}

\begin{remark}
This inequality indicates that the information quantity flowing in channel $C_1$ from system $S_1$ to system $S_2$ is lower bounded by two quantities characterized by the directed information. This inequality is true due to the transmission of the external input $x_0$, of which the quantity $I(y^{n};x_0|e^{n-1})$ is not presented in the right hand side of the inequality. This can be seen from the proof of this theorem.
\end{remark}

\begin{proof}
Recall the information identity from \cite{Limal_infoIdentity},
\begin{equation*}
I(x^n\rightarrow y^n) = I(x_0;y^n) + I(e^{n-1};x_0|y^n)+ I(e^{n-1}\rightarrow y^n)
\end{equation*}
Then,
\begin{equation*}
\begin{split}
 I(x_0;y^n) + I(e^{n-1};x_0|y^n)=I(x_0; (y^n, e^{n-1}) =  I(x_0;e^{n-1}) + I(y^{n};x_0|e^{n-1})\\
\end{split}
\end{equation*}
Using Lemma \ref{lem01}, the above equality is equivalent to
\begin{equation*}
\begin{split}
& I(x_0;y^n) + I(e^{n-1};x_0|y^n)= I(x^{n-1}\rightarrow e^{n-1}) + I(y^{n};x_0|e^{n-1})\\
\end{split}
\end{equation*}
Therefore, we have
\begin{equation*}
\begin{split}
I(x^n\rightarrow y^n) = & I(x_0;y^n) + I(e^{n-1};x_0|y^n)+ I(e^{n-1}\rightarrow y^n)\\
= & I(x^{n-1}\rightarrow e^{n-1}) + I(y^{n};x_0|e^{n-1}) +  I(e^{n-1}\rightarrow y^n)\\
\geq & I(x^{n-1}\rightarrow e^{n-1}) + I(e^{n-1}\rightarrow y^n).\\
\end{split}
\end{equation*}
\end{proof}

\begin{theorem}
Consider a feedback system as shown in Fig. \ref{Fig:general_feedback_system}. Let $m = x_0$, the information quantity flowing through system $S_1$ (from system inputs $e^{n-1}$ to system outputs $x^n$) can be decomposed as
\begin{equation*}
I(e^{n-1}\rightarrow x^n) = I(y^{n-1}\rightarrow x^{n}) + I(e^{n}\rightarrow x^n|| y^{n-1}).
\end{equation*}
\end{theorem}

\begin{proof}
\begin{equation*}
\begin{split}
 I(e^{n-1}\rightarrow x^n)=&\sum_{i=1}^{n}I(e^{i-1}, x_i|x^{i-1})\\
=&\sum_{i=1}^{n}H(x_i|x^{i-1})-H(x_i|x^{i-1}, e^{i-1})\\
=&\sum_{i=1}^{n}H(x_i|x^{i-1})- H(x_i|x^{i-1}, y^{i-1}) + H(x_i|x^{i-1}, y^{i-1}) - H(x_i|x^{i-1}, e^{i-1})\\
\stackrel{(a)}= &\sum_{i=1}^{n}H(x_i|x^{i-1})- H(x_i|x^{i-1}, y^{i-1}) + H(x_i|x^{i-1}, y^{i-1}) - H(x_i|x^{i-1}, e^{i-1}, y^{i-1})\\
= & \sum_{i=1}^{n} I(x_i; y^{i-1}|x^{i-1}) + I(x_i; e^{i-1}|x^{i-1}, y^{i-1})\\
= & I(y^{n-1} \rightarrow x^n) + I(e^{n-1} \rightarrow x^n || y^{n-1})\\
\end{split}
\end{equation*}
\end{proof}

\begin{remark}
The first quantity $I(y^{n-1}\rightarrow x^{n})$ can be interpreted as the information provided by the uncertainty of channel $C_1$, and similarly the second quantity $I(e^{n}\rightarrow x^n|| y^{n-1})$ can be interpreted as the information provided by the uncertainty of channel $C_2$. In fact, there are three external inputs into the feedback system in total. As one of them $x_0$ is known by the system $S_1$, the information flowing through $S_1$ intuitively should equal the information provided by the other two external inputs, i.e., noises in channel $C_1$ and $C_2$.
\end{remark}

\section{Conclusion}
In this paper, we derive three information decomposition identities and inequalities in general feedback systems. These decompositions indicate a law of conservation of information flows in closed-loop systems. These decompositions are beneficial in interpreting information flows of closed-loop systems, and serve as fundamental tools to derive insightful results when specified to particular feedback control/communicaiton problems.

\bibliographystyle{IEEEtran}
\bibliography{ref}

\end{document}